\documentclass[]{article}
\usepackage{amsthm}
\usepackage{amsmath}
\usepackage[utf8]{inputenc}
\usepackage[english, russian]{babel}
\usepackage{listings}
\usepackage{graphicx} 
\usepackage{booktabs}
\usepackage{algorithm}
\usepackage{algpseudocode}

\theoremstyle{definition}
\newtheorem{ad}{Definitions}
\theoremstyle{plain}
\newtheorem{blub}{Informal Theorem}
\newtheorem{blob}{Theorem}

\title{One Way Function Candidate based on the Collatz Problem}
\author{\CYRR\cyra\cyrd\cyre\,\CYRV\cyru\cyrch\cyrk\cyro\cyrv\cyra\cyrc\,(Rade Vuckovac)}
\date{\vspace{-5ex}}

\begin{document}

\maketitle
\selectlanguage{english}

\begin{abstract}
  The one way function based on the Collatz problem is proposed. It is based on the problem's conditional branching structure which is not considered as important even when the $3x+1$ question is quite famous. The analysis shows why the problem is mathematically so inaccessible and how the algorithm conditional branching structure can be used to construct one way functions. It also shows exponential dependence between algorithm's conditional branching and running cost of algorithm branch-less reductions.
\end{abstract}

\section*{Introduction}

According to Levin \cite{Levin:2003:TOF:772153.772176} the existence of One Way Function (OWF) is arguably the most important question in computing theory. OWF existence would imply $P \neq NP$ and the existence of some very important constructs in cryptography. For example: pseudorandom number generators, pseudorandom functions and various cryptographic protocols. Informally OWF is easy to compute given input $x$. When function description and output $y$ is given, however, it is difficult to guess input $x$.

In the absence of theoretically proven OWF, OWF candidates are used in practice. The bulk of asymmetrical encryption is based on OWF candidates such as factoring and discrete logarithm problems. It is not proven that they are reversible, but so far no efficient and non quantum algorithms are found yet. For example: if a relatively large composite integer is given, it is difficult to decompose it into a product of two integers (integer factorization).
	
	The proposed OWF candidate is based on the $3x+1$ problem. This problem is also known as the Collatz problem and it is a famous problem in mathematics. The problem is easy to state: an input is a positive integer, \emph{if} the input is even, divide it with $2$ \emph{otherwise} multiply it with $3$ and add $1$. The outcome is a new input and the same is repeated ... The conjecture, states that for every positive integer, the procedure will reach $1$ eventually. While the problem does not suffer lack of attention,  no single mathematical structure is associated with it and an arbitrarily chosen iteration behaves as a fairly flipped coin \cite{3x+1}. Therefore, the use of the $3x+1$ problem in cryptography is not surprising. Apple Inc applied for a patent using Collatz conjecture as a system and method for a hash function \cite{ciet2013system}.
	
	 The Wolfram's rule 30 cellular automata \cite{wolfram1984computation} is another example with branching structure. It has almost identical transformation procedure as the $3x+1$ problem. Its English formulation is \cite{Wolfram2002}: "First, look at each cell and its right-hand neighbor. \emph{If} both of these were white on the previous step, then take the new color of the cell to be whatever the previous color of its left-hand neighbor was. \emph{Otherwise}, take the new color to be the opposite of that" (emphasis added). The rule 30 is used as a pseudorandom number generator in Wolfram's Mathematica software. Both concepts rely either on the problem difficulty (Apple's Collatz based hash) or on the empirical evidence (Wolfram's rule 30). 
	
	Generally, an algorithm is perceived as a well defined procedure, but that is not always the case. Conditional branching can undermine procedure definition and make algorithm behaviour quite unpredictable. The OWF proposal uses underlying conditional branching structure of the $3x+1$ transformation (as is in the rule 30) to argue function reversal difficulty. Acquired complexity depends on a number of conditional branching iterations $r$ and not on input size $n$.
	
	For example: Let input $x$ be a positive integer $512$ bits long. Then apply modified Collatz transformation: if $x$ is even, divide it with $2$ otherwise multiply it with $3$ add $1$ and divide it with $2$.  Repeat the procedure  $256$ times and record the latest output ($x_{256})$ (See Figure \ref{fig:collatz}). 
	
	\begin{figure}[!htbp]  
		\begin{center}  
			\includegraphics[height=1.25in,width=4.25in]{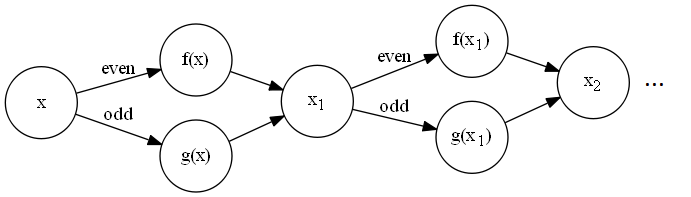}  
			\caption{\small \sl The $3x+1$ composite function; $f(x)=x/2$ and $g(x)=(3x+1)/2$.
				\label{fig:collatz}}  
		\end{center}  
	\end{figure}

Then split $512$ bit input $x$ to two $256$ bit values $xl$ (left part) and $xr$ (right part). Do the same for the latest iteration point $x_{256}$ resulting with $xl_{256}$ and $xr_{256}$. The path value $p$ is $256$ bit path record where "even" route step is $0$ and "odd" $1$ The output is calculated as bellow where $\oplus$ is exclusive or (basically output is xor of input, stop value and parity):
	\[y=xl\oplus xr\oplus xl_{256}\oplus xr_{256}\oplus{p}\]
	
Reversal difficulty lies in the absence of a function description. Without input specified, the transformation is in an ambiguous state and could happen in any of $2^{256}$ ways. Therefore \emph{only particular input defines corresponding transformation}. If an arbitrarily $256$ bit integer is presented as an output $y$, it is hard to prove if it is a valid output (let alone to find matching input $x$) thanks to function description $2^{256}$- sized ambiguous state.

The whole argumentation is based on the algorithm structure and classical notions of polynomial and exponential costs.

\begin{blub}
	To have a proper function description ($3x+1$ for example),the exponential nature of conditional branching must be circumvented in the algorithm implementing that function. Avoidance results in either exhaustive search with accompanying exponential running cost or the conditional branching is replaced with combination of sequence and iteration structure incurring polynomial cost only.
\end{blub}
  The analysis of relations between the running costs and the conditional branching complexity of various $3x+1$ algorithm variants is included in this paper. 
	
The rest of the paper is organised as follows: In section \ref{Prelim} two $3x+1$ algorithms are presented. One is using branching which looks exactly as Figure \ref{fig:collatz}. The other does not use branching at all. In section \ref{owf} the $3x+1$ candidate for one way function is shown and discussed.
	The section \ref{obstacle} contains discussion on the execution paths coverage and its relation to the costs and function description of the algorithms presented.

\section{Preliminaries}\label{Prelim}

In this paper, the argumentation assumes structured program theorem \cite{bohm1966flow} to be true. It states: a control flow graph needs only three structures to compute any computable function. They are: sequence, selection and iteration. The paper assumes branching structure as an elementary construction. Although it could be replaced with another two structures, the replacement can not happen in polynomial running time. That is discussed in section \ref{obstacle}.
	
	The cyclomatic complexity (CC) is also an important part of the following discussion. CC is a software metric which measures the amount of execution paths the program can take through execution. It is really counting predicates (algorithm branching) where every count doubles the amount of execution paths (if the decision is binary) \cite{mccabe1982structured}. The assumption is: if the algorithm has high CC it is in a state where the algorithm can be run but the functional description and the behaviour of it is unknown. For example, if a program has a CC more than 10 (meaning $2^{10}$ execution paths) the program should be rewritten because testing each path of that program becomes very costly and it is questionable what that program had in mind in the first place. Please see NIST article \cite{watson1996structured} for the algorithm CC recommendations. On the other hand, some algorithms have high CC and apparently it can not be avoided. For example the $3x+1$ problem with every iteration doubles the amount of possible execution paths. The same is for some cellular automata such as Wolfram's rule 30 \cite{wolfram1984computation}.
	
	To set a stage for the discussion, the $3x+1$ problem is presented with two algorithms. One with high CC and other with constant CC with respect to the number of iterations. Finally the $3x+1$ OWF is presented in section \ref{owf}.
	All three algorithms are based on the modified $3x+1$ function. The reason is to avoid extra iteration step consisting of dividing even part $3x+1$ with $2$.
	\[
		f(x) = \left\{\begin{array}{lr}
			x/2 & \text{if}\quad x\equiv 0 \\
			(3x+1)/2 & \text{if}\quad x\equiv 1
			\end{array}\right. (mod\enskip2)
	\]

	\subsection{3x+1 with Conditional Branching} 
		Let input $x$ be defined as a positive integer for the composite function and $r$ as a number of function iterations i.e. a number of functions $f(x)$ or $g(x)$ involved in composition:
		
		\[ \left\{ \begin{array}{ll}
		f & \mbox{if $x$ is even};\\
		g & \mbox{if $x$ is odd};
		\end{array} \right\}\circ
		\left\{ \begin{array}{ll}
		f & \mbox{if $x$ is even};\\
		g & \mbox{if $x$ is odd};
		\end{array} \right\}\circ\dots\mbox{$r$ times}	
		\]		
		Where:
		$f(x) = x/2$ is chosen when $x$ is even and
		$g(x) = (3x+1)/2$ is chosen when $x$ is odd. Pseudo code looks like:
		\begin{algorithm}[H]
			\caption{$3x+1$ algorithm}\label{colatz_normal}
			\begin{algorithmic}[1]
				\Procedure{Colatz}{$x,r$}\Comment{starting integer $x$ and iterations $r$}
				\For{\texttt{$i=0;\,i<r;\,i++$}}
				\If{$x$ is even}
				\State $x \gets x/2$
				\Else
				\State $x \gets (3x+1)/2$
				\EndIf
				\EndFor
				\State \textbf{return} $x$\Comment{finishing integer (output)}
				\EndProcedure
			\end{algorithmic}
		\end{algorithm}
		For example, $x = 3$ and $r = 2$ gives following:\\
		
		\begin{tabular}[c]{c c c}
			\toprule
			input 	& transformations 	& output\\
			\midrule
			3		& $g\circ g$		& 8\\
			\bottomrule
		\end{tabular}\\

\subsection{3x+1 Exhaustive Search}
		This algorithm and notations are equivalent with \ref{colatz_normal}. The difference is the structure of the algorithm. While \ref{colatz_normal} uses if/else for function composition, this algorithm uses exhaustive search to find a particular composition to match an input. First, depending on $r$ the list of all combination for composition is created:\\
		
		\begin{tabular}[c]{ccc}
			\toprule
			$r = 2$	& $r = 3$ & $r = \dots$ \\
			\midrule 
			$ f \circ f $	& $ f \circ f \circ f $  & $\dots$ \\ 
			$ f \circ g $	& $ f \circ f \circ g $ &  \\ 
			$ g \circ g $	& $ f \circ g \circ f $ &  \\ 
			$ g \circ f $	& $ f \circ g \circ g $ &  \\ 
			& $ g \circ g \circ g $ &  \\ 
			& $ g \circ g \circ f $ &  \\ 
			& $ g \circ f \circ g $ &  \\
			& $ g \circ f \circ f $ &  \\
			\bottomrule
		\end{tabular}\\
	
		Then, the algorithm tries an input with every composition from the corresponding list and stops if the composite function outputs a whole number.
		Pseudo code is:
		\begin{algorithm}[H]
			\caption{$3x+1$ search algorithm}\label{colatz_ex}
			\begin{algorithmic}[1]
				\Procedure{Colatz}{$x,r$}\Comment{starting integer $x$ and iterations $r$}
				\State \textbf{initialise} $y$ as rational and $i=0$	\Comment $i$ is a counter
				\State \textbf{initialise} $f(x)=x/2$ and $g(x)=(3x+1/2)$
				\State \textbf{initialise} list $l$		\Comment $2^r$ sized 2d array with all combination of $f$ and $g$
				\While{$y$ is rational}  
				\State $y \gets i_{th}\,l\, compositition$		\Comment do ith row from the list
				\State $i++$
				\EndWhile
				\State \textbf{return} $y$\Comment{finishing integer (output)}
				\EndProcedure
			\end{algorithmic}
		\end{algorithm}
		For example, $x = 3$ and $r = 2$. The algorithm will process the first column from previous table $(r=2)$:\\
		\begin{tabular}[l]{ll}
			\midrule
			$ f \circ f\, (3)=3/4$ & result not integer,  try next entry\\
			$ f \circ g\, (3)=2\; 3/4$ & result not integer, try next entry\\
			$ g \circ g\, (3)=8$ & result integer, it will stop and output 8\\
			\bottomrule
		\end{tabular}\\

	\section{One way function}\label{owf}
	
	It is not the first time when the Collatz conjecture is used in a cryptographic application. Apple Inc. applied for a patent using $3x+1$ problem as a hash system and method \cite{ciet2013system}. However approaches from Apple and the proposed OWF are different. Apple uses traditional $3x+1$ transformations (Input and Output column Table \ref{3x+1encoding}) and add more operations. Please see quote below:
	
	\begin{quotation}
	1. A method comprising:
	receiving an input value and an iteration value;
	based on the iteration value, iteratively performing steps comprising:
	if a least significant bit of the input value is 0,\\
	(1) dividing the input value by a first value\\
	if a least significant bit of the input value is 1,\\
	(1) multiplying the input value by a second value,\\
	(2) adding one to the input value, and\\
	(3) applying a modulo operation of a prime value to the input value, to yield a first iteration value,
	to yield an updated input value;
	returning the updated input value as a hash value.
	\end{quotation}
	Note: It appears that the "first value" and "second value" could be numbers other than $2$ and $3$ respectively.
	
	In contrast, the proposed OWF takes exclusive or of the input value, the last iterated value and execution path encoding as an output. Using the example from Table \ref{3x+1encoding}, OWF output is calculated as:
	\[y=9\oplus 13\oplus 46\] 
	
	\subsection{$3x+1$ OWF}
		This algorithm is also based on $3x+1$ problem.
		The difference from $3x+1$ problem algorithm is in a way how the algorithm stops and what is the  actual output. Variables are defined below:  
		\begin{itemize}
			\item Input $x$ is positive integer $512$ bits long (binary encoding) $l_x=512$
			\item The number of iterations $r=256$
			\item The algorithm output is a binary string $ y $ with the same bit length as $ r $
			\item $ s $ is a record of selection decisions through the algorithm execution
		\end{itemize}
		
		The program takes $ x $ as an input, runs as $ 3x+1 $ algorithm with $ r $ iterations and output  $ y $. Pseudo code is:\\
		\begin{algorithm}[H]
			\caption{$3x+1$ OWF algorithm}\label{colatz_owf}
			\begin{algorithmic}[1]
				\Procedure{Colatz OWF}{$x,s,r$}\Comment{$512$ bit $x$, $256$ bit $s$ and iterations $r=256$}
				\State \texttt{create $xl$ and $xr$}	\Comment{equally divided sides of $x$ ($256$ bit)}
				\For{\texttt{$i=0;\,i<r;\,i++$}}
				\If{$x$ is even}
				\State $x = x/2$
				\State $s_i \gets 0$			\Comment{ith bit of $s$ becomes $0$}
				\Else
				\State $x = (3x+1)/2$
				\State $s_i \gets 1$			\Comment{ith bit of $s$ becomes $1$}
				\EndIf
				\EndFor
				\State \texttt{create $x'l$ and $x'r$}	\Comment{divided sides of final $x$ ($x'r\,256$ bit wide)}
				\State $y=xl\oplus xr\oplus x'l\oplus x'r\oplus s$ 	\Comment {$\oplus$ exclusive or}
				\State \textbf{return} $y$\Comment{finishing integer (output)}
				\EndProcedure
			\end{algorithmic}
		\end{algorithm}

	\begin{table}[!htbp]
		\centering
		\caption{$3x+1$ OWF for $x=9$ and $r=6$, $y=9 \oplus 13 \oplus 46$}
		\label{3x+1encoding}
		\begin{tabular}{ccccc}
			\toprule
			Step $r$ & Input & Function & Output & Path Encoding  \\ 
			\midrule
			1  & starts with \textbf{9}    & $(3x+1)/2$    & 14   		          & 1        \\ 
			2  & 14                       & $(x/2)$       & 7                     & 0        \\ 
			3  & 7                        & $(3x+1)/2$    & 11                    & 1        \\ 
			4  & 11                       & $(3x+1)/2$    & 17                    & 1        \\ 
			5  & 17                       & $(3x+1)/2$    & 26                    & 1        \\ 
			6  & 26                       & $(x/2)$       & ends with \textbf{13}  & 0        \\ 
			\bottomrule
			&&&& *decimal value \textbf{46}
		\end{tabular}
	\end{table} 

\section{Discussion on the Branching Obstacle}\label{obstacle}
	
	The execution cost of algorithm \ref{colatz_normal} has linear dependency on the number of iterations $r$. The second algorithm \ref{colatz_ex} is exhaustive search and the cost has exponential dependency on number of iterations $r$ because of $2^r$-sized table used for search.
	The relation of execution paths for the above algorithms are the opposite. For algorithm costs and number of paths relation see Table \ref{costing}. When examined further, both algorithms have problems:
	\begin{itemize}
		\item The problem with the branching algorithm is not so obvious. It runs fine with linear cost w.t.r. of the number of iterations. The problem starts when input is not specifically defined. For example, if $r=256$ and unknown input $x>r$ in length, the $3x+1$ algorithm can take any of $2^{256}$ possible execution paths, practically $2^{256}$ different composite functions could be applied. That is quite opposite of what is expected from a mathematical function with well defined  description. For example $sin$ function is properly specified; if $\sin30=0.5$ and $\sin60=0.866$ then $\sin$ of angles between $30$ and $60$ are somewhere in range of $0.5$ and $0.866$. That means for a range of angles function output behaviour is quite predictable. The same reasoning can not be applied to $3x+1$ algorithm because there is no underlying transformation defined. To have proper $3x+1$ function description (for example algorithm \ref{colatz_owf}) someone has to go through the all inputs ($2^{512}$) and create an ordered table of all inputs and outputs which is not practical for large $r$. 
		\item Non branching algorithms do have well defined execution paths but running costs is exponential with respect to the number of iterations $r$.
   \end{itemize}

\begin{table}[!htbp]
		\centering
		\caption{Costs and path coverage with respect to iterations $r$}
		\label{costing}
		\begin{tabular}{lll}
			\toprule
			 $3x+1$ algorithm variants & running cost & number of paths\\
			\midrule
		 	 branching algorithm \ref{colatz_normal}      & polynomial       	& $2^r$        	\\
			 non branching algorithm \ref{colatz_ex}  & $O(2^{r/2})$     	& constant     	\\ 
			 mirage						 		& polynomial      	& constant     	\\ 
			\bottomrule
		\end{tabular}
	\end{table} 

Lets consider set $V$ containing all the $3x+1$ algorithm variants (or reductions). CC is used to categorise all variants. Using CC is beneficial  because all variants have a certain programming structure and consequently CC metrics assigned. The list of all variants can be divided in two groups according to the algorithm associated CC:
	\begin{itemize}
		\item The subset $C$ where $C \subset	V$ is a branching algorithm subset where CC depends exponentially on iteration ($r$) (1st row Table \ref{costing}). 
		\item The remaining variants are in the second subset $R$ where $R \subset	V$ and $C+R=V$ (2nd row Table \ref{costing}).
	\end{itemize}
	The subset $R$ can be divided again into two subsets using algorithm running cost:
	\begin{itemize}
		\item The subset $E$ with exponential running cost w.r.t. iterations $r$ (exhaustive search) the same as algorithm \ref{colatz_ex} Table \ref{costing} where $E \subset	R$. 
		\item The subset $G$ of remaining variants where running cost is polynomial w.r.t. iteration $m$ (mirage in Table \ref{costing}) where $G \subset R$ and $E+G=R$. Subset $G$ is interesting because its members have polynomial running cost and execution paths are well defined.
	\end{itemize}
	
	Mirage variant (3rd row Table \ref{costing}) is a member of subset $G$. If mirage variant exists, such an algorithm will behave in a similar fashion as a non branching algorithm \ref{colatz_ex} but instead of checking all entries in the $2^r$-sized table it will always choose adequate function composition. Essentially it will take input $x$ and the number of iterations $r$ and will execute $3x+1$ in polynomial time with respect to $r$ without using branching programming structure. Sure enough, mirage reassembles non deterministic polynomial (NP) algorithm from complexity theory where it always chooses the correct path when a branching decision is needed. 
	
	From all above, Theorem \ref{theo} can be deduced. Note that structured program theorem \cite{bohm1966flow} already implies conditional branching as a basic algorithmic structure. The reason for this addition is to clarify cases when conditional branching is replaced with sequence and iteration structures (for example Algorithm \ref{colatz_ex}). 
	
	\begin{ad}
		The list of definition is:
	\begin{itemize}
			\item 	
	 $cb$; Conditional branching is an algorithm structure which causes different sequence execution depending on some comparison. One example is modified Colatz \textit{if else} statement.
		\[
		f(x) = \left\{\begin{array}{lr}
		x/2 & \text{if}\quad x\,even \\
		(3x+1)/2 & \text{otherwise}
		\end{array}\right.
		\]
		\item
	$si$; Sequence and iteration reduction is a situation when the $cb$ is replaced with other two algorithmic structures. For example, algorithms \ref{colatz_normal} and \ref{colatz_ex} have different structures ($cb$ and $si$) but on the same inputs produce identical outputs.
	\item 
	$r$; Iterations represent a number of $cb$ steps. One configuration example is shown in Figure \ref{fig:collatz}. In that case, the number of execution paths is $2^r$ where $2$ is binary branching and $r$ is number of steps. This case is the simplest case for the path counting. Other $bc$ configurations are determined by CC metric procedure \cite{mccabe1982structured} and resulting CC value $c$ is equivalent to $r$ (path count is $2^c$).
	\item 
	$exp$; Set of all algorithms with exponential costs.
	\end{itemize}
	\end{ad}

	\begin{blob}\label{theo}
		There exists at least one conditional branching structure $cb$ from the set of all possible conditional branching structures $CB$. It can not be replaced with equivalent sequence/iteration structure $si$ in algorithm $a_{si}$ and at the same time have polynomial cost for that replacement.
		\[
		\exists cb \in CB:cb=si \wedge \forall a_{si} \in exp
		\] 
	\end{blob}
	
	\begin{proof}
		The opposite of the theorem statement is assumed (that none of $cb$ exists). Consequently, every possible conditional branching construct $cb$ will have equivalent sequence/iteration combination $si$ with polynomial cost reduction. In other words, every program or algorithm can be constructed with sequences and iterations only.
	\end{proof}
	
	The value of the above discussion is in categorisation applicability to every possible branching scenario. If all conditional branching cases are analysed and if it is found that every branching case has an accompanying mirage variant then branching programming structure is \emph{redundant}. Otherwise, designing algorithms with high cyclomatic complexity and asking difficult questions without specifying input is quite easy.

\bibliographystyle{unsrt}
\bibliography{owfcp} 

\end{document}